\documentclass[10pt,a4paper]{article}
\usepackage[latin1]{inputenc}
\usepackage{algorithm}
\usepackage{placeins}
\usepackage{algpseudocode}
\usepackage{amsmath}
\usepackage{amsthm}
\usepackage{amsfonts}
\usepackage{amssymb}
\usepackage{enumerate}
\usepackage{xcolor}
\usepackage[hidelinks]{hyperref}
\hypersetup{colorlinks,linkcolor={red!50!black},citecolor={blue!50!black},urlcolor={green!50!black}}

\newcommand{\rank}{\mathbf{rank}}
\renewcommand{\d}{\operatorname{d}}
\renewcommand{\c}{\mathbf{c}}
\newcommand{\e}{\mathbf{e}}
\newcommand{\Mat}{\mathbf{Mat}_{n\times n}(\F)}
\renewcommand{\L}[1]{\mathcal{L}_{#1}[\Fn]}
\newcommand{\F}{\mathbb{F}_q}
\newcommand{\Fn}{{\mathbb{F}_{q^n}}}

\renewcommand{\H}{\mathbf{H}}

\newcommand{\U}{\mathbf{U}}
\newcommand{\C}{\mathcal{C}}
\newcommand{\n}{\eta}
\newcommand{\m}{\mu}
\newcommand{\M}{\mathbf{M}}
\newcommand{\A}{\mathbf{A}}
\newcommand{\B}{\mathbf{B}}
\newcommand{\0}{\mathbf{0}}
\renewcommand{\ll}{\pmb{\mathbf{\lambda}}}

\renewcommand{\l}{\lambda}
\newcommand{\Tr}{\mathbf{Tr}\,}

\newcommand{\W}{\mathbf{W}}
\renewcommand{\a}{\alpha}

\theoremstyle{plain}
\newtheorem{thm}{Theorem}
\newtheorem{lem}{Lemma}

\newtheorem{cor}{Corollary}
\theoremstyle{definition}
\newtheorem{defn}{Definition}
\theoremstyle{remark}
\newtheorem{rem}{Remark}
\newtheorem{exa}{Example}

\author{Tovohery Hajatiana Randrianarisoa\footnote{The author is supported by SNF grant no. 169510}}
\title{A Decoding Algorithm for Rank Metric Codes}
\date{December 20, 2017} 

\begin{document}
\maketitle
\begin{abstract}
In this work we will present algorithms for decoding rank metric codes. First we will look at a new decoding algorithm for Gabidulin codes using the property of Dickson matrices corresponding to linearized polynomials. We will be using a Berlekamp-Massey-like algorithm in the process. We will show the difference between our and existing algorithms. Apart from being a new algorithm, it is also interesting that it can be modified to get a decoding algorithm for general twisted Gabidulin codes.
\end{abstract}

\section{Introduction}
Rank metric codes have many applications in network coding and in cryptography. So far, there are two known general constructions of rank metric codes with arbitrary parameters. The first class of rank metric codes are the Gabidulin codes \cite{Del78,Gab85,Ksh05} and they were generalised to the twisted Gabidulin codes \cite{She16,Lun15}. Several decoding algorithms already exist for Gabidulin codes \cite{Gab85,Loi06,Ric04}. For twisted Gabidulin codes, a decoding algorithm exists but only for a particular parameters of the code \cite{Ran17}. Most of the algorithms for Gabidulin codes are using syndrome computation, extended Euclidean algorithm, Berlekamp-Massey algorithm. In this work we will use the Berlekamp-Massey algorithm for a rank metric code again, but in a different way. Namely, suppose that $(c_1,\cdots,c_n)+(e_1,\cdots,e_n)$ is the received vector with $(e_1,\cdots,e_n)$ being the error. We first interpolate the polynomial $f(x)+g(x)$ from the received vector, where $f(x)$ is the message polynomial and $g(x)$ is the polynomial corresponding to the error vector. Due to the form of $f(x)$ (it has degree $q^{k-1}$ at most), we already know some coefficients of $g(x)$ and we will show that these coefficients are enough to recover the whole polynomial $g(x)$. This algorithm can be further modified to be used with all cases of twisted Gabidulin codes. To do these, we will first give a description of Gabidulin codes and twisted Gabidulin codes in Section \ref{sec:2}. Also, we will give a brief description of two decoding algorithms for Gabidulin codes. Then, in Section \ref{sec:3}, we prove a theorem which enables us to build a new decoding algorithm. We will present the new decoding algorithm for Gabidulin codes and we will show its difference with the two algorithms we presented in Section \ref{sec:2}. In Section \ref{sec:4}, we will modify the algorithm to use it with twisted Gabidulin codes. Finally, we will conclude in Section \ref{sec:5}.

\section{Rank metric codes}\label{sec:2}
\begin{defn}\label{defn:1}
Let $\Fn/\F$ be a finite field extension of degree $n$. A linearized polynomial of $q$-degree $k$ is a polynomial of the form
\[
f(x) = f_0 x + f_1 x^q +\cdots + f_{k} x^{q^{k}},
\]
where $f_k\neq 0$ and $f_i\in\Fn$ for any integer $i$ with $0\leq i\leq k$. The set of all these polynomials will be denoted by $\L{}$. If we fix an integer $k$, then $\L{k}$ denotes the set of all linearized polynomials of $q$-degree at most $k-1$.
\end{defn}
\begin{exa}\label{exa:1}
The trace map $\Fn\rightarrow\F$ is the linearized polynomial $\Tr x = x + x^q+ \cdots + x^{q^{n-2}}+x^{q^{n-1}}$
\end{exa}

We know that $\Fn$ is a vector space over $\F$ of dimension $n$. And since $x^{q^i}$ is an $\F$-automorphism of $\Fn$, we see that any linearized polynomial $f(x)\in \L{n}$ is an $\F$-linear map $\Fn\rightarrow\Fn$. In fact, $\L{n}$ is isomorphic to the set of all $(n\times n)$-matrices over $\F$. In this regard, we define the rank of a linearized polynomial to be its rank as an $\F$-linear map on $\Fn$.

Furthermore, if $f(x)$ has $q$-degree $k<n$, then considering it as a polynomial in $\Fn[x]$, it can have $q^k$ roots at most. Therefore  as an $\F$-linear map $\Fn\rightarrow\Fn$, $f(x)$ has a kernel of dimension of $k$ at most. This property allows us to use these polynomials to construct rank metric codes with good property.

For more on the theory of linearized polynomials, one can have a look at \cite{Lid96} Chapter 3.

\begin{defn}\label{defn:2}
Given a finite field $\F$, a rank metric code is a subset $\C$ of $\Mat$ together with the metric defined by $\d(\c_1, \c_2)=\rank(c_1-c_2)$, for $\c_i\in \Mat$. We call it linear if $\C$ is a vector space over $\F$ of dimension $k$ and if furthermore, $d$ is the minimum distance between two distinct codewords of $\C$, then we say that $\C$ is a $[n\times n,k,d]$ linear rank metric code.
\end{defn}

\begin{rem}\label{rem:1}
Alternative representations of a rank metric codes are given by the following:
\begin{enumerate}
\item The code is $\C\subset (\Fn)^n$, where the $\rank$ of a matrix is replaced by the maximum number of $\F$-linearly independent elements in $(c_1,\cdots,c_n)\in(\Fn)^n$. 
\item The code is $\C\subset (\L{n})$, where the $\rank$ of a matrix is replaced by the rank of the linearized polynomial $c(x)$ as an $\F$-linear map.
\end{enumerate} 
\end{rem}

Given a rank metric code and a minimum distance, the upper-bound on the size of the code is given by the following theorem.

\begin{thm}[\cite{Del78}]\label{thm:1}
Let $\C$ be a linear code in $\Mat$, if the minimum distance of $\C$ is equal to a positive integer $d$, then $\sharp \C \leq q^{n(n-d+1)}$. Codes for which the bound is attained (i.e $[n\times n, n(n-d+1), d]$) are called maximum rank distance (MRD) codes.
\end{thm}

The first class of linear rank metric code is the family of Gabidulin codes \cite{Del78,Gab85}, given by, 
\[
\C = \L{k} =\left\lbrace a_0 x + \cdots + a_{k-1} x^{q^{k-1}}, a_i\in \Fn \right\rbrace.
\]

To prove that this is MRD, we just use the fact that the kernel has dimension $k-1$ at most and then use the rank nullity theorem.

This construction was generalized by Sheekey in \cite{She16}. Namely, the Class of twisted Gabidulin codes is defined as follows:
\[
\C'_{\n,r} = \left\lbrace a_0 x + \cdots + a_{k-1} x^{q^{k-1}}+\n a_0^{q^r}x^{q^k}, a_i\in \Fn \right\rbrace,
\]
where $\n\in\Fn$ with $\prod_{i=0}^{n-1}\n^{q^i}\neq (-1)^{nk}$ and $r$ is a non-negative integer. This is an MRD code because a codeword cannot have $q^k$ zeroes by the choice of the the coefficients of $x$ and $x^{q^k}$ \cite{She16}.

\begin{rem}\label{rem:2}
The above representations are using linearized polynomials. These are $\F$-linear maps $\Fn\rightarrow \Fn$. To get a representation of the code as a subset of $(\Fn)^n$, we evaluate the code on a fixed basis of the extension $\Fn/\F$. Using this basis, we can also get a representation in the matrix form in $\Mat$.
\end{rem}

\begin{rem}\label{rem:3}
The construction can be generalized by replacing the monomials $x^{q^i}$ by $x^{q^{si}}$ where $(s,n)=1$. They are called generalized Gabidulin codes, see \cite{Ksh05}.
\end{rem}

There are already a lot of decoding algorithms for Gabidulin codes. For twisted Gabidulin codes, there is a decoding algorithm but only for some specific parameters \cite{Ran17}. As we mentioned before, we will give another decoding algorithm for Gabidulin codes and we will show how to modify it to get a decoding algorithm for twisted general Gabidulin codes. In order to see the difference between the existing and our algorithms, we will first show a brief description of two decoding algorithms for an $[n,kn,d]$-Gabidulin codes 
\begin{enumerate}
\item Compute the syndrome vector $r \H^T$, where $\H$ is a parity check matrix of the code given by $h_{i,j} = h_j^{q^i}$. The entries $s_i$ of this vector define a linearized polynomial $S(x)$.
\item Determine two linearized polynomials $L(x)$ and $F(x)$ such that $F(x) = L(x)\circ S(x) \mod z^{q^{d-1}}$. Here, there are two methods: Use Berlekamp Massey \cite{Ric04} or use the extended Euclidean algorithm \cite{Gab85}.
\item Find a basis $\lbrace e_i\rbrace$ of the kernel of $L(x)$.
\item Compute the $\a_i$'s from $\sum_i e_i \a_i^{q^j} = s_j$.
\item Find a matrix $Y$, with $\a_j = \sum_i Y_{j,i} h_i$.
\item Finally the error vector is $eY$, where the entries of $e$ are $e_i$.
\item output the message as $r-e$.
\end{enumerate}

\section{Decoding algorithm for Gabidulin codes}\label{sec:3}
Before we give our new decoding algorithm, we first give the needed tool. We know that linear maps can be decomposed as a sum of several linear maps of dimension one. And this can be shown in the setting of linearized polynomials. In the remaining part of this paper, we will use only linearized polynomials in $\L{n}$.
\begin{lem}\label{lem:1}
Any $\F$-linear map $\Fn\rightarrow \F$ can be represented by $\Tr(ax)$ for a fixed $a\in \Fn$.
\end{lem}
\begin{proof}
The set of $\Tr(ax)$ are obviously $\F$-linear maps $\Fn\rightarrow \F$. The equality comes by looking at the dimension of the space of $\F$-linear maps $\Fn\rightarrow \F$.
\end{proof}
\begin{cor}\label{cor:1}
Let $\l$ be an element of $\Fn$ and let $\l\F$ be the $\F$-subspace of $\Fn$ generated by $\l$. Then any linear map $\Fn\rightarrow \l\F$ has the form $\l\Tr(ax)$ for some $a\in \Fn$.
\end{cor}
The above representation in the corollary  is of course not unique.
As a consequence of the previous corollary, we have the following theorem.

\begin{thm}\label{thm:2}
Let $f$ be a linearized polynomial of rank $r$, then there are two subsets of $\Fn$ $S_1 = \left\lbrace a_1,\cdots , a_r\right\rbrace$ and $S_2 = \left\lbrace b_1,\cdots , b_r\right\rbrace$ such that they are both linearly independent over $\F$ and that
\[
f(x) =  a_1\Tr(b_1 x) + \cdots  a_r\Tr(b_r x).
\]
\end{thm}
\begin{proof}
Since $f$ is of rank $r$, then we choose $\left\lbrace a_1,\cdots , a_r\right\rbrace$ to be a generator of the image of $f$ as a linear map. By Corollary \ref{cor:1}, each projection of $f$ onto the subspace $\left< a_i\right>$ has the form $a_i\Tr(b_i x)$. Thus we get the desired form of $f$. What remains to show is the linear independence of the $b_i$'s. Without loss of generality, say $b_1 =\m_2 b_2 + \cdots + \m_r b_r$, with $\m_i\in\F$. Then
\begin{align*}
f(x) & =  a_1\left((\Tr ((\m_2 b_2 + \cdots + \m_r b_r)x)\right)) + a_2\Tr(b_2 x) + \cdots + a_r\Tr(b_r x) \\
& = a_1\Tr (\m_2 b_2 x) + \cdots + a_1\Tr (\m_r b_rx) + a_2\Tr(b_2 x) + \cdots +  a_r\Tr(b_r x) \\
& =(a_2+a_1\m_2)\Tr(b_2 x) + \cdots +  (a_r+a_1\m_r)\Tr(b_r x).
\end{align*}
Thus rank of $f$ is at most $r-1$ which is a contradiction.
\end{proof}

From Theorem \ref{thm:2}, we get the following corollary.

\begin{cor}\label{cor:2}
Let $f(x)$ be a linearized polynomial of rank $r$ over the field extension $\Fn/\F$ such that
\[
f(x) = f_0 x + f_1 x^q +\cdots + f_{n-1} x^{q^{n-1}}.
\]
Then there are two subsets of $\Fn$ $S_1 = \left\lbrace a_1,\cdots , a_r\right\rbrace$ and $S_2 = \left\lbrace b_1,\cdots , b_r\right\rbrace$ such that they are both linearly independent over $\F$, and  for all integer $i$ such that $0\leq i\leq n-1$,
\[
f_i =  \sum_{j=1}^r b_j^{q^i} a_j.
\]
\end{cor}

\begin{defn}\label{defn:3}
Let $f(x) = f_0 x + f_1 x^q +\cdots + f_{n-1} x^{q^{n-1}}$ be a linearized polynomial. The Dickson matrix associated to $f(x)$ is the matrix 
\[
\M = \begin{pmatrix}
f_0 & f_{n-1}^q & \hdots & f_1^{q^{n-1}} \\ 
f_1 & f_{0}^q & \hdots & f_2^{q^{n-1}} \\ 
\vdots & \vdots & \ddots & \vdots \\
f_{n-1} & f_{n-2}^q & \hdots & f_0^{q^{n-1}} \\ 
\end{pmatrix}.
\]
\end{defn}

Another matrix related to linearized polynomials is the Moore matrix. 

\begin{defn}
Given $\left\lbrace a_1,\cdots, a_k\right\rbrace\subset \Fn$, the Moore matrix associated to the $a_i$'s is the matrix
\[
\begin{pmatrix}
a_1 & a_2 & \hdots & a_k \\
a_1^q & a_2^q & \hdots & a_k^q \\
\vdots & \vdots & \ddots & \vdots \\
a_1^{q^{k-1}} & a_2^{q^{k-1}} & \hdots & a_k^{q^{k-1}} \\
\end{pmatrix}
\]
\end{defn}
It is well known that the above Moore matrix is invertible if and only if the $a_i$'s are linearly independent over $\F$.

As a consequence of Corollary \ref{cor:2}, we have the following theorem.

\begin{thm}\label{thm:3}
Let $f(x)$ be a linearized polynomial of rank $r$ over the field extension $\Fn/\F$ such that
\[
f(x) = f_0 x + f_1 x^q +\cdots + f_{n-1} x^{q^{n-1}}.
\]

Let $\M_1,\cdots,\M_n$ be the rows of the matrix $\M$ as in Definition \ref{defn:3}.

Then we have the following property:
\begin{enumerate}[(i)]
\item The matrix $\M$ is of rank $r$.
\item Any $r$ successive rows $\M_i,\cdots,\M_{i+r}$ are linearly independent and the other rows are linear combinations of them.
\item All $(r\times r)$-matrices $\left(M_{i,j}\right)_{(i \mod n,j \mod n)};\; l_1\leq i\leq l_1+r,\; l_2\leq j\leq l_2+r$ with $0\leq l_i\leq n-1$ are invertible.
\end{enumerate}
\end{thm}
\begin{proof}
From Corollary \ref{cor:2}, one sees that $\M=\B\A$, where
\[
\B= 
\begin{pmatrix}
b_1 & b_2 & \hdots & b_r \\
b_1^q & b_2^q & \hdots & b_r^q \\
\vdots & \vdots & \ddots & \vdots \\
b_1^{q^{n-1}} & b_2^{q^{n-1}} & \hdots & b_r^{q^{n-1}} \\
\end{pmatrix}
\text{ and }
\A=
\begin{pmatrix}
a_1 & a_1^q & \hdots & a_1^{q^{n-1}} \\
a_2 & a_2^q & \hdots & a_2^{q^{n-1}} \\
\vdots & \vdots & \ddots & \vdots \\
a_r & a_r^q & \hdots & a_r^{q^{n-1}} \\
\end{pmatrix}
\]
Since the $a_i$'s are linearly independent over $\F$ and the same for the $b_i's$, we see that each $r$ successive rows of $B$ and any $r$ successive columns of $A$ constitute invertible matrices. All statements of the theorem follow from these facts.
\end{proof}

It is this theorem that is important to us. This enables us to build a new decoding algorithm.

We are now ready to explain the decoding algorithm. It consists of two steps. The first part is to interpolate the received message to construct the polynomial $f(x)+g(x)$, where $f(x)$ is the message polynomial and $g(x)$ is the error polynomial. Since $f(x)$ is of degree $k-1$ at most, we should know the coefficient of $x^{q^i}$ in $g(x)$, $\forall i\geq k$. We will show that these coefficients are actually enough to recover the whole polynomial $g(x)$ with some condition on the rank of $g(x)$.

\subsection{Polynomial interpolation}

First of all, depending on the representation of the code, we need to do some interpolation to get a linearized polynomial form. Assume that our encoding was given by 
\begin{align*}
\L{k} &\rightarrow (\Fn)^n \\
f(x) &\mapsto (c_1,c_2,\cdots,c_n),
\end{align*}
where $\left\lbrace a_1,a_2,\cdots,a_n\right\rbrace$ is a fixed basis of $\Fn/\F$ and $c_i = f(a_i)$.

We assume that an error of $\e =(e_1,e_2,\cdots,e_n)$ was added to the original codeword and suppose that $\rank(e) = t < \frac{n-k+1}{2}$. Therefore $(r_1,r_2,\cdots,r_n)$ was received with $r_i = c_i + e_i$.

Let $\U$ be the Moore matrix
\[
\U =
\begin{pmatrix}
a_1 & a_1^q & \hdots & a_1^{q^{n-1}} \\
a_2 & a_2^q & \hdots & a_2^{q^{n-1}} \\
\vdots & \vdots & \ddots & \vdots \\
a_n & a_n^q & \hdots & a_n^{q^{n-1}} \\
\end{pmatrix}
\]
Then
\[
\U\begin{pmatrix}
f_0 + g_0 \\
\vdots \\
f_{n-1} + g_{n-1}
\end{pmatrix}
=
\begin{pmatrix}
r_1 \\
\vdots\\
r_n
\end{pmatrix}
\]
where $g(x) = \sum g_i x^{q^i}$ is the error polynomial corresponding to $\e$ i.e. $g(a_i) = e_i$. Obviously, $g(x)$ as an $\F$-linear map has rank $t< \frac{n-k+1}{2}$.

Thus, we may compute $\U^{-1}$ in advance and then compute
\[
\U^{-1}\begin{pmatrix}
r_1 \\
\vdots\\
r_n
\end{pmatrix}.
\]
This gives us $f_0+g_0,\cdots, f_{n-1}+g_{n-1}$. Since $f_i = 0, \forall i\geq k$, we now know the values of $g_k,\cdots g_{n-1}$. In the next step, we will use these coefficients to recover the other coefficients of $g(x)$.

\subsection{Polynomial reconstruction}

Let us have a look at the matrix $\M$ in Theorem \ref{thm:3} from the error polynomial $g(x)$. We consider its submatrix
\[
\W = \begin{pmatrix}
g_0 & g_{n-1}^q & \hdots & g_{k+t-1}^{q^{n-(k+t-1)}} & \hdots & g_k^{q^{n-k}} & \hdots & g_1^{q^{n-1}} \\ 
g_1 & g_{0}^q & \hdots & g_{k+t}^{q^{n-(k+t-1)}} & \hdots & g_{k+1}^{q^{n-k}} & \hdots & g_2^{q^{n-1}} \\ 
\vdots & \vdots & \ddots & \vdots & \ddots & \vdots & \ddots & \vdots \\
g_{t-1} & g_{t-2}^q & \hdots & g_{k+2t-2}^{q^{n-(k+t-1)}} & \hdots & g_{k+t-1}^{q^{n-k}} & \hdots & g_t^{q^{n-1}} \\ 
g_t & g_{t-1}^q & \hdots & g_{k+2t-1}^{q^{n-(k+t-1)}} & \hdots & g_{k+t}^{q^{n-k}} & \hdots & g_{t+1}^{q^{n-1}}
\end{pmatrix}.
\]
We know that the $t$ last rows are linearly independent and that the first row should be a linear combination of the $t$ last rows. 

This gives us an equation of the form
\begin{equation}\label{eqn:1}
(\l_0,\cdots,\l_t)  \begin{pmatrix}
g_{k+t-1}^{q^{n-(k+t-1)}} & \hdots & g_k^{q^{n-k}} \\ 
g_{k+t}^{q^{n-(k+t-1)}} & \hdots & g_{k+1}^{q^{n-k}} \\ 
\vdots & \ddots & \vdots \\
g_{k+2t-2}^{q^{n-(k+t-1)}} & \hdots & g_{k+t-1}^{q^{n-k}}\\ 
g_{k+2t-1}^{q^{n-(k+t-1)}} & \hdots & g_{k+t}^{q^{n-k}}
\end{pmatrix} = \0
\end{equation}
where we may assume that $\l_0=1$.

Notice that by the interpolation step, we know the coefficients $g_k,\cdots,g_{n-1}$. Since we suppose that $t<\frac{n-k+1}{2}$, then $k+2t-1\leq n-1$. Thus we know all the coefficients $g_k,\cdots,g_{k+2t-1}$. And thus, by Theorem \ref{thm:3}, this equation has unique solution in $\ll$ which we can compute. This can be done for example by using matrix inversion but that will take $\mathcal{O}(t^3)$ operations. 

Similarly to the case of Reed-Solomon codes, we can do better. Namely, we have here a Toeplitz-like matrix. And this can actually be solved by using a Berlekamp-Massey-like algorithm from \cite{Ric04}. To see this let $u_i = g_i^{q^{n-i}}$. Therefore, Equation \eqref{eqn:1} becomes

\begin{equation}\label{eqn:2}
(\l_0,\cdots,\l_t)  \begin{pmatrix}
u_{k+t-1} & \hdots & u_k \\ 
u_{k+t}^{q} & \hdots & u_{k+1}^{q} \\ 
\vdots & \ddots & \vdots \\
u_{k+2t-2}^{q^{t-1}} & \hdots & u_{k+t-1}^{q^{t-1}}\\ 
u_{k+2t-1}^{q^{t}} & \hdots & u_{k+t}^{q^{t}}
\end{pmatrix} = \0
\end{equation}

We want to find $\l_1,\cdots,\l_t$ from the sequence $(u_{k+2t-1}, \cdots, u_{k+t}, \cdots, u_k)$. Equation \eqref{eqn:2} is exactly the form of recurrence shown in \cite{Ric04}. In that paper, they gave an algorithm for solving Equation \eqref{eqn:2}. We will give the algorithm in Algorithm \ref{algo:1}. We set $d=n-k+1$.

\begin{algorithm}[ht!]
\caption{Berlekamp-Massey}\label{algo:1}
\begin{algorithmic}[1]
\Procedure{BERLEKAMP-MASSEY}{$s_0,\cdots,s_{2t-1}$}
\State $L\gets 0$
\State $\Lambda^{(0)}(x) \gets x$
\State $B^{(0)}(x) \gets x$
\State $i \gets 0$
\While{$i\leq d-2$}
\State $\Delta_i \gets s_i + \sum_{j=1}^L\l_j^{(i)}s_{i-j}^{q^j}$
\State $\Lambda^{(i+1)} \gets \Lambda^{(i)}-\Delta_i x^{q}\circ B^{(i)}(x)$
\If{$\Delta_i == 0$}
\State $B^{(i+1)}(x) \gets x^q\circ B^{(i)}(x)$
\Else
\If{$2L>i$}
\State $B^{(i+1)}(x) \gets x^q\circ B^{(i)}(x)$
\Else
\State $B^{(i+1)}(x) \gets\Delta_i^{-1}\Lambda^{(i)(x)}$
\State $L \gets i+1-L$
\EndIf
\EndIf
\State $i\gets i+1$
\EndWhile
\State \Return $\Lambda^{(i)}(x)$
\EndProcedure
\end{algorithmic}
\end{algorithm}

In this algorithm $\Lambda^{(i)}(x)=\sum_j\l^{(i)}_j x^{q^j}$ and at the end of the algorithm, we will just collect the coefficient of the $\Lambda^{(i)}(x)$ to get our $\l_i$. Notice that on input we take $(s_0,\cdots,s_{2t-1}) = (u_{k+2t-1},\cdots,u_k)$.

We summarize our decoding algorithm with the following steps in Algorithm \ref{algo:2}. Suppose $(r_1,r_2,\cdots,r_n)$ was received with an error of rank $t<\frac{n-k+1}{2}$. We already know the matrix $\U^{-1}$ in advance.

\begin{algorithm}[ht!]\caption{Decoding algorithm}\label{algo:2}
Input: $(r_1,\cdots,r_n)$
\begin{enumerate}[(1)]
\item Compute
\[
\begin{pmatrix}
f_0+g_0 \\
\vdots \\
f_{k-1}+g_{k-1} \\
g_k \\
\vdots \\
g_{n-1}
\end{pmatrix}
= \U^{-1}
\begin{pmatrix}
r_1 \\
\vdots \\
r_n
\end{pmatrix}
\]
\item Use the Berlekamp-Massey-like Algorithm \ref{algo:1} to get the $\l_i$'s.
\item Use the fact that the first row of the matrix $\W$ is a linear combination of the remaining rows, using the $\l_i$'s, to recursively compute the remaining coefficients of $g(x)$. This is just like recursively computing elements of a sequence but the difference with linear-feedback shift register is that the steps also involves raising to some power of $q$.
\item Output  the message as $(f+g)(x)-g(x)$.
\end{enumerate}
\end{algorithm}
\subsection{Complexity and comparison with other algorithms}
All the three first steps of Algorithm \ref{algo:2} have quadratic complexity i.e they can be done in $\mathcal{O}(n^2)$ operations in $\Fn$. The last step is a linear operation. Thus in general we have an algorithm with $\mathcal{O}(n^2)$ operations in $\Fn$.

We already saw two decoding algorithms in Section \ref{sec:2}. As we can see, there is a difference in the first steps of these algorithms and our algorithm. Instead of using an $((n-k)\times n)$ matrix for computing the syndromes, we use an $(n\times n)$ matrix to interpolate $f+g$. So in the first step, we have some extra $\mathcal{O}(nk)$ extra multiplications. The second steps are more or less the same as they are either Berlekamp-Massey or extended Eulidean algorithm. The last steps are where we may get the advantage as we directly use a linear recurrence to recover the error polynomial. For the other algorithms in Section \ref{sec:2}, one first needs to compute the roots of some polynomials (error locator polynomial) before one can reconstruct the error vectors using some relations. 

\section{Extension to twisted Gabidulin codes}\label{sec:4}
In this section, we will explain that our algorithm can also be modified to get a decoding algorithm for twisted Gabidulin codes. And in contrary to the algorithm in \cite{Ran17}, we can do it for any parameters. We assume that the original message was given by
\[
f(x) = f_0 x +\cdots + f_{k-1}x^{q^{k-1}}+\n f_0^{q^r}x^{q^k}.
\]
After the interpolation step, we get the polynomial $f(x)+g(x)$. In opposite to the case of Gabidulin codes, we do not know the value of $g_k$ from this. However the problem we are faced remains similar. We want to find a linear relations between the rows of 
\[
\W = \begin{pmatrix}
g_0 & g_{n-1}^q & \hdots & g_{k+t-1}^{q^{n-(k+t-1)}} & \hdots & g_k^{q^{n-k}} & \hdots & g_1^{q^{n-1}} \\ 
g_1 & g_{0}^q & \hdots & g_{k+t}^{q^{n-(k+t-1)}} & \hdots & g_{k+1}^{q^{n-k}} & \hdots & g_2^{q^{n-1}} \\ 
\vdots & \vdots & \ddots & \vdots & \ddots & \vdots & \ddots & \vdots \\
g_{t-1} & g_{t-2}^q & \hdots & g_{k+2t-2}^{q^{n-(k+t-1)}} & \hdots & g_{k+t-1}^{q^{n-k}} & \hdots & g_t^{q^{n-1}} \\ 
g_t & g_{t-1}^q & \hdots & g_{k+2t-1}^{q^{n-(k+t-1)}} & \hdots & g_{k+t}^{q^{n-k}} & \hdots & g_{t+1}^{q^{n-1}}
\end{pmatrix}.
\]
where we know the values $g_{k+1},\cdots,g_{n-1}$ and $\n g_0^{q^r}-g^{q^k}$. We will see that we still can solve this problem.  We have an equation of the form
\begin{equation*}
(\l_0,\cdots,\l_t)  \begin{pmatrix}
g_{k+t}^{q^{n-(k+t)}} & g_{k+t-1}^{q^{n-(k+t-1)}} & \hdots & g_k^{q^{n-k}} \\ 
g_{k+t+1}^{q^{n-(k+t)}} & g_{k+t}^{q^{n-(k+t-1)}} & \hdots & g_{k+1}^{q^{n-k}} \\ 
\vdots & \ddots & \vdots \\
g_{k+2t-1}^{q^{n-(k+t)}} & g_{k+2t-2}^{q^{n-(k+t-1)}} & \hdots & g_{k+t-1}^{q^{n-k}}\\ 
g_{k+2t}^{q^{n-(k+t-1)}} & g_{k+2t-1}^{q^{n-(k+t-1)}} & \hdots & g_{k+t}^{q^{n-k}}
\end{pmatrix}
 = \0.
\end{equation*}
Notice that we introduce one more columns in the equation. Again, by assumption, we have $2t<n-k+1$. Thus $k+2t\leq n$. If $k+2t<n$, then a Berlekamp-Massey algorithm using the columns of the previous matrix except the last column is enough to compute the $\l_i$'s. If $k+2t=n$, then the equation becomes,
\begin{equation}\label{eqn:3}
(\l_0,\cdots,\l_t)  \begin{pmatrix}
g_{k+t}^{q^{n-(k+t)}} & g_{k+t-1}^{q^{n-(k+t-1)}} & \hdots & g_k^{q^{n-k}} \\ 
g_{k+t+1}^{q^{n-(k+t)}} & g_{k+t}^{q^{n-(k+t-1)}} & \hdots & g_{k+1}^{q^{n-k}} \\ 
\vdots & \ddots & \vdots \\
g_{n-1}^{q^{n-(k+t)}} & g_{k+2t-2}^{q^{n-(k+t-1)}} & \hdots & g_{k+t-1}^{q^{n-k}}\\ 
g_{0}^{q^{n-(k+t)}} & g_{n-1}^{q^{n-(k+t-1)}} & \hdots & g_{k+t}^{q^{n-k}}
\end{pmatrix}
=\0,
\end{equation}
where two entries in terms of $g_{k}$ and $g_{0}$ are unknown.

If we use the columns of the matrix except the first and last columns, then we should have an underdetermined system of linear equations whose solution space is of dimension two. We assume that two linearly independent solutions are $\ll$ and $\ll'$. They can be found using the Berlekamp-Massey like algorithm again. Thus a solution of equation \eqref{eqn:3} is of the form $\ll+A\ll'$ for some $A\in\Fn$.
Using this with the first column and the last column, we get two equations. Furthermore, we also know $\n g_0^{q^r}-g_k$. So in total we get a system of three equations with three unknowns,
\begin{equation}\label{eqn:4}
\begin{cases}
\ h_0+h_1A+(h_2+h_3A)g_0^{q^{n-(k+t)}} &= 0 \\
\ h_4+h_5A+(h_6+h_7A)g_k^{q^{n-k}} &= 0 \\
\ h_8+\n g_0^{q^r}-g_k &= 0 
\end{cases}.
\end{equation}
In this system, we know all the $h_i$ and $g_0,g_k,A$ are unknown. Notice that any solution of the system of equation \eqref{eqn:4} is actually a solution of the decoding algorithm. By the unique decoding property, there can only be one solution of this system.

To solve the system, we use the third equation in the two first equations and we get
\begin{equation}\label{eqn:5}
\begin{cases}
\ s_0+s_1A+(s_2+s_3A)g_0^{q^{i}} &= 0 \\
\ s_4+s_5A+(s_6+s_7A)g_0^{q^{j}} &= 0 
\end{cases}
\end{equation}
with  the $s_i$'s known. We can further reduce this into one variable equation of the form, for some integer $l$,
\[
\frac{s_0+s_1A}{s_2+s_3A} = \frac{s_4'+s_5'A^{q^l}}{s_6'+s_7'A^{q^l}}.
\]
We want to point out that this form of equation was also obtained in \cite{Ran17}. However, in our case here, we are sure that any solution would give us the closest codeword to the received message. We have now reduced the problem to solving the polynomial equation of the form
\[
P(A) = u_0 + u_1A +u_2A^{q^l}+A^{q^l + 1} = 0.
\]
We distinguish three cases:
\begin{itemize}
\item If $u_0 = u_1u_2$, then we can factor $P(A) = (A^{q^l}+u_1)(A+u_2)$.
\item If $u_1=u_2^q$, then 
\begin{align*}
P(A) &= u_0 + u_2^{q^l}A +u_2A^{q^l}+A^{q^l + 1}\\
&= u_0 -u_2^{q^l}u_2+(A+u_2)^{q^l+1}
\end{align*}
\item If $u_0 \neq u_1u_2$ and $u_1\neq u_2^q$, then, from \cite{Blu04}, by a change of variable $y=(u_2u_1-u_0)(u_1-u_2^{q^l})^{-1}A-u_2$, we will get a polynomial equation of the form 
\[
Q(y) = y^{q^l+1}-vy+v = 0
\]
with $v=(u_1-u_2^{q^l})^{q^l+1}/(u_0-u_2u_1)^{q^l}$.
\end{itemize}

First of all, it is easy to show that if we get $A$ from $P(A)$, then we can use equation \eqref{eqn:5} to get $g_0$. And we use equation \eqref{eqn:4} to get $g_k$. These will give us the error polynomial $g(x)$ with the recurrence relation from equation \eqref{eqn:3}. So, normally, there should be only one unique solution for $A$. Now the question is how do we solve the equation $P(A)=0$? Any of the three cases which produce multiple solutions should be ruled out. The first case of $P(A)$ is easy to solve. The two last cases reduce to polynomials of the form
\[
P(X) =X^{q^l+1}+aX+b. 
\]
The number of roots of such polynomials was studied in \cite{Blu04}. Here we will give a method to find these roots. 

Suppose that $y_2$ is a root of $P(X)$. Then set $b=-y_2^{q^l+1}-ay_2$ and choose $y_1=-a-y_2^{q^l}$. thus $b=y_2y_1$. We get
\begin{align*}
(x^{q^l} - y_1x)\circ(x^{q^l} - y_2x)& = x^{q^{2l}}-y_2^{q^l}x^{q^l}-y_1x^{q^l}+y_1y_2x\\
&=x^{q^{2l}}+ax^q+bx.
\end{align*}
The converse is also true. So, to get the root of $P(X)$, we just need to factor the linearized polynomial $x^{q^{2l}}+ax^q+bx$. In case this polynomial admits a root $x_0$ in $\Fn$ then we just take $y_2=x_0^{q-1}$. Otherwise, we will need to use a factorization algorithm like in \cite{Gie98}.

Once $A$ is computed, we can compute $g_0$ and $g_k$. Then we continue the decoding algorithm with the same methods as with the Gabidulin codes.

\begin{rem}\label{rem:4}
These algorithms can be easily modified to get a decoding algorithm for generalized (twisted) Gabidulin codes. Namely instead of working with the field automorphism $x^q$, we work with automorphisms of the form $x^{q^s}$. 
\end{rem}

\section{Conclusion}\label{sec:5}
In this work we have given a new decoding algorithm for Gabidulin codes. First, instead of computing syndromes, we do some polynomial interpolation. Our algorithm requires more computations in this first steps but we can compensate this in the last steps. Namely, there is no need to find roots of some ``error locator polynomial''. We just need to use a recurrence relation to recover the ``error polynomial'' after using a Berlekamp-Massey-like algorithm. We gave a brief analysis on the complexity and a comparison of our algorithm to some existing decoding algorithms for Gabidulin codes. Furthermore,  we show that our algorithm can be modified to get a general decoding algorithm for twisted Gabidulin codes. 

Finally, we think that it is possible to get a version of our algorithm for Reed-Solomon codes. Namely we can use an equivalent of the Dickson matrix. In the case of Reed-Solomon codes, we have a circulant matrix. And a theorem of K\"onig-Rados gives a relation between the number of non-zero roots of a polynomial and the rank of some circulant matrix, see \cite{Lid96}, Chapter 6, Section 1. It is known that the most expensive steps in the decoding of Reed-Solomon codes is finding roots of the error locator polynomials. This can be avoided in our algorithm.

We have seen that our algorithm involves factoring linearized polynomial of degree $2$. It is well known that factoring a regular polynomial of degree $2$ can be done by computing the discriminant of the polynomial. The algorithm presented in \cite{Gie98} gives a factorization for linearized polynomials of general degree, we could further simplify our algorithm if we would have a discriminant like method to factorize a degree $2$ linearized polynomial.
\section*{Aknowledgement}

I would like to thank Anna-Lena Horlemann-Trautmann and Joachim
Rosenthal for their valuable comments and suggestions on this work.

\FloatBarrier
\bibliography{reference}

\end{document}